\documentclass[11pt]{amsart}

\usepackage[english]{babel}
\usepackage[utf8]{inputenc}

\usepackage{a4wide}

\usepackage{amsfonts}
\usepackage{amsmath}
\usepackage{amsthm}
\usepackage{mathtools}
\mathtoolsset{showonlyrefs}
\usepackage{amssymb}

\usepackage{graphicx}
\usepackage{subfigure}
\usepackage[justification=centering]{caption}

\usepackage{enumerate}

\usepackage[title]{appendix}

\usepackage{hyperref}
\usepackage{url}

\usepackage[foot]{amsaddr}

\providecommand{\U}[1]{\protect\rule{.1in}{.1in}}
\newtheorem{theorem}{Theorem}
\newtheorem*{acknowledgement}{Acknowledgements}

\newtheorem{conjecture}[theorem]{Conjecture}
\newtheorem{corollary}[theorem]{Corollary}

\newtheorem{definition}[theorem]{Definition}

\newtheorem{lemma}[theorem]{Lemma}

\newtheorem{proposition}[theorem]{Proposition}

\renewenvironment{proof}[1][Proof]{\noindent\textbf{#1.} }{\ \rule{0.5em}{0.5em}}

\providecommand{\norm}[1]{\lVert#1\rVert}

\newcommand{\R}{\mathbb{R}}

\newcommand{\Rn}{\mathbb{R}^n}
\newcommand{\Z}{\mathbb{Z}}

\newcommand{\Lt}[1][n]{L^2 \left(\R^{#1}\right)}

\newcommand{\G}{\mathcal{G}}

\newcommand{\F}{\mathcal{F}}

\newcommand{\indicator}{\raisebox{2pt}{$\chi$}}

\renewcommand{\l}{\lambda}
\renewcommand{\L}{\Lambda}

\newcommand{\J}{
	\left(
		\begin{array}{rc}
		0 & I\\
		-I & 0
		\end{array}
	\right)
}

\newcommand{\vol}{\textnormal{vol}}

\begin{document}

\title{Gaussian Distributions and Phase Space Weyl--Heisenberg Frames}
\author[M.~Faulhuber]{Markus Faulhuber}
\email{markus.faulhuber@univie.ac.at}
\author[M.A.~de Gosson]{Maurice A.~de Gosson}
\email{maurice.de.gosson@univie.ac.at}
\author[D.~Rottensteiner]{David Rottensteiner}
\email{david.rottensteiner@univie.ac.at}
\address[M.F., M.G., D.R.]{NuHAG, Faculty of Mathematics, University of Vienna\newline Oskar-Morgenstern-Platz 1, 1090 Vienna, Austria}
\address[M.F.]{Analysis Group, Department of Mathematical Sciences, NTNU Trondheim\newline Sentralbygg 2, Gløshaugen, Trondheim, Norway}



%
%
%
\begin{abstract}
	Gaussian states are at the heart of quantum mechanics and play an essential role in quantum information processing. In this paper we provide approximation formulas for the expansion of a general Gaussian symbol in terms of elementary Gaussian functions. For this purpose we introduce the notion of a ``phase space frame" associated with a Weyl-Heisenberg frame. Our results give explicit formulas for approximating general Gaussian symbols in phase space by phase space shifted standard Gaussians as well as explicit error estimates and the asymptotic behavior of the approximation.
\end{abstract}


\maketitle

\section{Introduction and Main Result}

Gaussian distribution functions play a central role in many areas of mathematics ranging from statistics to signal theory. This is not only because of their relative simplicity, but also because of their usefullness in many applied areas. For instance, in signal processing Gaussian smoothing, the blurring of an image by a Gaussian function, is used to reduce noise (theory of low-pass filters). In quantum mechanics the so-called Gaussian states are fundamental in quantum communication protocols \cite{ARL}, one of the reasons being that such states and their evolution are more accessible in the laboratory than their non-Gaussian counterparts \cite{se17}. The present paper is motivated by finding convenient expansions of distributions of the type
\begin{equation}\label{Gauss1}
	\rho(z) = \sqrt{\det\Sigma^{-1}} \, e^{-\pi \, \Sigma^{-1}z^{2}}
\end{equation}
where $\Sigma$ is a real symmetric positive matrix of order $2n\times2n$ and the vector $z=(x,p)$ is an element of $\R^{2n}\equiv \R^{n} \times \R^{n}$. Since the function $\rho$ is even and $L^{1}$-normalized, i.e.
\begin{equation}
	\int_{\R^{2n}} \rho(z)\, d^{2n}z=1,
\end{equation}
the matrix $\Sigma$ is the covariance matrix of $\rho$ viewed as a centered normal probability distribution:
\begin{equation}
	\Sigma = \int_{\R^{2n}} z\rho(z)z^{T} \, d^{2n}z \quad \mbox{ for } \quad z=\binom{x}{p}.
\end{equation}
In the quantum case, $\rho$ has the following interpretation: assume that the eigenvalues of the Hermitian matrix $\Sigma+\frac{i\hbar}{2}J$ are all non-negative ($J$ is the standard symplectic matrix); then the Weyl operator with symbol $(2\pi\hbar)^{n}\rho$ is the quantum mechanical density operator \cite{Birkbis}
\begin{equation}
	\widehat{\rho}=\sum_{j}\lambda_{j}\widehat{\rho}_{j},
\end{equation}
where $(\widehat{\rho}_{j})_{j}$ is a sequence of mutually orthogonal projectors on normalized vectors $\psi_{j}\in L^{2}(\mathbb{R}^{n})$ and $(\lambda_{j})$ a sequence of positive numbers summing up to one; it follows that we have
\begin{equation}\label{wigsum}
	\rho=\sum_{j}\lambda_{j}W\psi_{j},
\end{equation}
where $W\psi_{j}$ is the usual Wigner transform of $\psi_{j}$ defined by
\begin{equation}
	W\psi_{j}(x,p)=\left(  \tfrac{1}{2\pi\hbar}\right)^{n} \int_{\R^{n}}e^{-\frac{i}{\hbar}py}\psi_{j}(x+\tfrac{1}{2}y)\overline{\psi_{j}(x-\tfrac{1}{2}y)} \, d^{n}y.
\end{equation}
Now, formula \eqref{wigsum} raises the following non-trivial question: given a Gaussian distribution \eqref{Gauss1}, is it possible to write it as a linear combination of elementary Gaussians? In the quantum case this amounts to asking whether the functions $\psi_{j}$ could themselves be chosen as simpler Gaussians. In general, this problem is still open, however for the special cases of $2 \times 2$ and $2 \times 3$ systems a solution is known \cite{Horo_1996} and, to some extent, numerical methods for higher dimensions are available \cite{Lein_2006}.

The aim of this article is to study in some detail the decomposition of Gaussian mixed states into pure Gaussian states using the theory of Weyl--Heisenberg frames (also called Gabor frames). In general, these systems are non-orthogonal and overcomplete and, hence, neither the expansion coefficients are uniquely determined nor is their sum equal to 1. However, there is a canonical choice for the coefficients and it is possible to control their $\ell^2$-norm by the frame inequality \eqref{eq_frame} given further down. We reformulate the notion of Weyl--Heisenberg frames as in \cite{ACHA} using the Weyl--Wigner--Moyal formalism. This has the advantage of making the underlying symplectic covariance properties, which play an essential role in the study of Gaussians, more obvious. We will give explicit formulas; these turn out to be rather complicated, but are certainly of use in applications, both theoretical and numerical. Our main results are as follows.

\begin{theorem}\label{thm_main}
	For $x \in \Rn$ we denote the $n$-dimensional standard Gaussian by
	\begin{equation}
		\varphi^\hbar(x) = (\tfrac{1}{\pi \hbar})^{n/4} e^{-\tfrac{1}{2 \hbar} x^2}
	\end{equation}
	Let $\L = \delta^{-1} Q S \Z^{2n} \subset \R^{2n}$ be a lattice with $\delta > 0$, $Q \in SO(2n,\R) \cap Sp(n)$ and
	\begin{equation}
		S =
		\begin{pmatrix}
			L & 0\\
			L^{-T} P & L^{-T}
		\end{pmatrix}
		\in Sp(n),
	\end{equation}
	for which $L$ is invertible and $P = P^T$, whose associated Weyl-Heisenberg system $\G(\varphi^\hbar, \L)$ is a frame for $\Lt$. Consider the proper subspace $\mathcal{H}_{\varphi^\hbar} \subset \Lt[2n]$ given as the image of the linear operator
	\begin{align}
		U_{\varphi^\hbar}: \Lt &\to \Lt[2n], \\
		 \psi &\mapsto \Psi = (2 \pi \hbar)^{n/2} W(\psi,\varphi^\hbar),
	\end{align}
	where $W(\psi,\varphi^\hbar)$ is the cross-Wigner transform. Furthermore, we define
	\begin{equation}
		\widetilde{T}(z_\l) \Phi^{\hbar} = (2 \pi \hbar)^{n/2} W(\widehat{T}(z_\l) \varphi^\hbar,\varphi^\hbar),
	\end{equation}
	for the usual Heisenberg operator $\widehat{T}$. Then, the system
	\begin{equation}
		\widetilde{\G}(\Phi^\hbar, \L) = \{ \widetilde{T}(z_\l) \Phi^{\hbar}: z_\l \in \L \}
	\end{equation}
	is a (so-called phase space) frame for $\mathcal{H}_{\varphi^{\hbar}}$. For the associated (phase space) frame operator $\widetilde{A}_\G$ we obtain 	the stable approximation
	\begin{equation}\label{eq_convergence_thm}
		|| Id - \delta^{-2n} \hspace{2pt} \widetilde{A}_\mathcal{G} ||_{op} = \mathcal{O}\Bigl( e^{-\frac{\delta^2}{4 \hbar} ( ||L||^2_{\R^{n \times n}} + ||L||^{-2}_{\R^{n \times n}})} \Bigr),
	\end{equation}
	as the parameter $\delta \to \infty$.
\end{theorem}

\begin{theorem}\label{thm_main2}
	Under the assumptions of Theorem \ref{thm_main}, for $\Psi \in \mathcal{H}_{\varphi^\hbar}$ the (phase space) frame expansion approximates the element $\Psi$ in the sense of
	\begin{equation}\label{eq_approx_thm}
		\Psi \approx \delta^{-2n} \widetilde{A}_\mathcal{G} \Psi = \delta^{-2n} \sum_{z_\l \in \L} (( \Psi \mid \widetilde{T}(z_\l) \Phi^\hbar)) \hspace{2pt} \widetilde{T}(z_\l) \Phi^\hbar,
	\end{equation}
	where the accuracy of the approximation is determined by \eqref{eq_convergence_thm}.
	
	In particular, for a generalized $n$-dimensional Gaussian defined by
	\begin{equation}
		\varphi_M^\hbar (x) = (\tfrac{1}{\pi \hbar})^{n/4} \det(Re(M))^{1/4} e^{-\tfrac{1}{2 \hbar} M x^2},
	\end{equation}
	with $M = M^*$ and $Re(M) > 0$, we have
	\begin{equation}
		\Phi_F^\hbar (z) = U_{\varphi^\hbar} \varphi_M^\hbar (z) = \det(Re(M))^{1/4} \det(\tfrac{1}{2} (M+I))^{-1/2} \hspace{2pt} e^{-\tfrac{1}{\hbar} F z^2} \in \mathcal{H}_{\varphi^\hbar},
	\end{equation}
	with
	\begin{equation}
		F = 
		\begin{pmatrix}
			2 (M+I)^{-1} M & -i(M-I)(M+I)^{-1}\\
			-i(M+I)^{-1} (M-I)^{-1} & 2 (M+I)^{-1}
		\end{pmatrix}.
	\end{equation}
	We obtain the numerically stable approximation
	\begin{equation}
		\Phi_F^\hbar \approx \delta^{-2n} \sum_{z_\l \in \L} c_{z_\l} \hspace{2pt} \widetilde{T}(z_\l) \Phi^\hbar
	\end{equation}
	with convergence rate~\eqref{eq_convergence_thm} from Theorem~\ref{thm_main} and coefficients
	\begin{equation}
		c_{z_\l}
		= \bigl( \tfrac{2}{\pi \hbar} \bigr)^{n/2} \hspace{2pt} \det(F+I)^{-1/2} \hspace{2pt} e^{-\frac{1}{4 \hbar} z_\l^2}
		\hspace{2pt} e^{-\tfrac{1}{4 \hbar} (F+I)^{-1} \bigl( (J -iI) z_\l \bigr)^2}.
	\end{equation}
\end{theorem}

The paper is structured as follows:
\begin{itemize}
	\item In Section \ref{sec_WH} we first review the main results we will need from the theory of Weyl--Heisenberg frames using the notation and approach in \cite{ACHA}; we thereafter propose an extension of this notion to phase space, where we make use of the notion of \textquotedblleft Bopp quantization\textquotedblright, which one of us has introduced and studied in \cite{gobopp1,gobopp2,Birkbis} (also see de Gosson and Luef \cite{golubopp}). The resulting frames, called phase space frames, are related to the usual Weyl--Heisenberg frames by a family of partial linear isometries $U_{\phi}: L^{2}(\mathbb{R}^{n})\longrightarrow L^{2}(\mathbb{R}^{2n})$ defined in terms of the cross-Wigner transform.
	
	\medskip
	
	\item In Section \ref{sec_Gauss_n} we study $n$-dimensional Gaussian mixed states. We begin by recalling results about the cross-Wigner and cross-ambiguity transforms of pairs of Gaussians. These formulas (which have their own intrinsic interest) are necessary for computing the Weyl--Heisenberg coefficients in a Gaussian frame expansion. We successively derive the results stated in Theorems~\ref{thm_main} and \ref{thm_main2}.
	
	\medskip
	
	\item In Section \ref{sec_Gauss_1} we restrict our study to Gaussian states in the case $n=1$. The reason for treating this case separately is that a full characterization of Gaussian Weyl--Heisenberg frames is known and $Sp(1) = SL(2,\R)$, which means that all lattices are symplectic. Also, the blocks appearing in the symplectic matrix $S$ are scalars in this case. Hence, some of the formulas simplify substantially.
\end{itemize}

Our notation and results are meant for mathematical physicists as well as for people working in time-frequency analysis. For the latter group, the dependence on $\hbar$ might seem irritating at the beginning, but by setting $\hbar = \tfrac{1}{2 \pi}$ one obtains the classical setting for time-frequency analysis.

\subsection{Notation and Terminology}

The generic point in phase space $\mathbb{R}^{2n}\equiv\mathbb{R}^{n}\times\mathbb{R}^{n}$ is denoted by $z=(x,p)$, where we have set $x=(x_{1},\dots,x_{n})$, $p=(p_{1},\dots,p_{n})$. The scalar product of two vectors $p$ and $x$ is denoted by $px$. When matrix calculations are performed, $z,x,p$ are viewed as column vectors. We write $M x^2$ for the quadratic form $x^T M x$, and $M xp$ for $p^T M x$. For an invertible matrix $M$ we write $M^{-T}$ for its transposed inverse. Moreover, we equip $\mathbb{R}^{2n}$ with the standard symplectic structure
\begin{equation}
	\sigma(z,z')=p x'-p' x,
\end{equation}
in matrix notation $\sigma(z,z')=(z')^{T}Jz$, where
$J = \J$ is the standard symplectic matrix. The symplectic group of $\mathbb{R}^{2n}$ is denoted by $\operatorname*{Sp}(n)$; it consists of all linear automorphisms of $\mathbb{R}^{2n}$ such that $\sigma(Sz,Sz')=\sigma(z,z')$ for all $z,z'\in\mathbb{R}^{2n}$. Working in the canonical basis, $\operatorname*{Sp}(n)$ is identified with the group of all real $2n \times 2n$ matrices $S$ such that $S^{T}JS=J$ (or, equivalently, $SJS^{T}=J$).

We will write $d^{2n}z=d^{n}x \hspace{2pt} d^{n}p$, where $d^{n}x=dx_{1} \dots dx_{n}$ and $d^{n}p=dp_{1} \dots dp_{n}$. The scalar product in $L^{2}(\mathbb{R}^{n})$ is denoted by
\begin{equation}
	(\psi | \phi)=\int_{\R^n}\psi(x)\overline{\phi(x)}d^{n}x
\end{equation}
and the associated norm by $|| \, . \, ||$ (in physicist's bra-ket notation we thus have $(\psi|\phi)=\langle\phi|\psi\rangle$). In phase space we denote the inner product by
\begin{equation}
	((\Psi | \Phi )) = \int_{\R^{2n}} \Psi(z) \overline{\Phi(z)} \, d^{2n}z
\end{equation}
and the induced norm by $|||\, . \,|||$.

The Schwartz space of rapidly decreasing functions is denoted by $\mathcal{S}(\mathbb{R}^{n})$ and its dual space by $\mathcal{S}'(\mathbb{R}^{n})$.
The Fourier transform on $\R^n$ is formally defined by
\begin{equation}
		\F \psi(p) = \left(\tfrac{1}{2\pi \hbar}\right)^{n/2} \int_{\R^n} \psi(x) e^{-\tfrac{i}{\hbar} p  x} \, d^nx.
	\end{equation}

\section{Weyl--Heisenberg Frames}\label{sec_WH}

\subsection{Definition and Terminology}

It is customary in frame theory to introduce Gabor frames, which we prefer to call Weyl--Heisenberg frames in this work because of their very close relationship with the theory of the Heisenberg group and Weyl pseudodifferential calculus.

Let $\phi$ be a (non-zero) square integrable function (the \textit{frame} \textit{window}) on ${\mathbb{R}}^{n}$, and let $\Lambda$ a discrete (countable) subset of ${\mathbb{R}}^{2n}$ (the index set of the frame). The associated \textit{Weyl--Heisenberg system} is the family of square-integrable functions
\begin{equation}
	{\mathcal{G}}(\phi,\Lambda)=\{\widehat{T}(z_{\lambda})\phi:z_{\lambda}\in\Lambda\},
\end{equation}
where $\widehat{T}(z)=e^{-i\sigma(\hat{z},z)/\hbar}$ is the Heisenberg operator and $\widehat{z} = (\widehat{x}, \widehat{p})$ is the formal position-momentum operator with $\widehat{x} \psi = x \psi$ and $\widehat{p} \psi = i \hbar \partial_t \psi$. The action of $\widehat{T}(z)$ in $\Lt[n]$ is given by
\begin{equation}\label{heiwe}
	\widehat{T}(z_{0})\phi(x)=e^{\tfrac{i}{\hbar}(p_{0}x-\frac{1}{2}p_{0}x_{0})}\phi(x-x_{0}),
\end{equation}
where $z_{0}=(x_{0},p_{0})$. They satisfy the commutation and addition properties
\begin{equation}
	\widehat{T}(z_{0}) \widehat{T}(z_{1}) = e^{\tfrac{i}{\hbar} \sigma(z_0,z_1)} \widehat{T}(z_{1}) \widehat{T}(z_{0})
\end{equation}
and
\begin{equation}
	\widehat{T}(z_{0} + z_1) = e^{-\tfrac{i}{2\hbar} \sigma(z_0,z_1)} \widehat{T}(z_{0}) \widehat{T}(z_{1}).
\end{equation}
See \cite{Birk,Birkbis,Littlejohn} for detailed studies of these operators.

We will call the system ${\mathcal{G}}(\phi,\Lambda)$ a \textit{Weyl--Heisenberg frame} if there exist constants $a,b>0$ (\textit{frame bounds}) such that
\begin{equation}\label{eq_frame}
	a||\psi||^{2}\leq \sum_{z_{\lambda}\in\Lambda}|(\psi \, | \, \widehat{T}(z_{\lambda})\phi)|^{2} \leq b||\psi||^{2}
\end{equation}
for every $\psi\in L^{2}({\mathbb{R}}^{n})$. If $a=b$, then the frame ${\mathcal{G}}(\phi,\Lambda)$ is said to be \emph{tight}. Note that a tight frame is an orthonormal basis if all frame elements of $\G(\phi,\L)$ are normalized and $a=b=1$ (see e.g.~\cite[Lemma 5.1.6.]{Gro01}). An immediate example of a Weyl--Heisenberg frame is given by the well-known Fourier (orthogonal) basis $\G(\indicator_{[0, 2 \pi \hbar)}, \Z \times \Z)$.

Given a Weyl--Heisenberg frame ${\mathcal{G}}(\phi,\Lambda)$, the associated frame operator $\widehat{A}_{{\mathcal{G}}}$ is given by
\begin{equation}\label{fop1}
	\widehat{A}_{{\mathcal{G}}}\psi=\sum_{z_{\lambda}\in\Lambda}(\psi \, | \, \widehat{T}(z_{\lambda})\phi) \hspace{2pt} \widehat{T}(z_{\lambda})\phi.
\end{equation}
It is a positive, bounded, self-adjoint, invertible operator on $L^{2}({\mathbb{R}}^{n})$ with bounded inverse, and we have
\begin{equation}\label{eq_frame_dual}
	\psi = \sum_{z_{\lambda}\in\Lambda} (\psi \, | \, \widehat{T}(z_{\lambda}) \widehat{A}_{{\mathcal{G}}}^{-1} \phi) \hspace{2pt} \widehat{T}(z_{\lambda})\phi
	= \sum_{z_{\lambda}\in\Lambda}(\psi \, | \, \widehat{T}(z_{\lambda})\phi) \hspace{2pt} \widehat{T}(z_{\lambda}) \widehat{A}_{{\mathcal{G}}}^{-1}\phi.
\end{equation}

We will see that Weyl--Heisenberg frames can be expressed, both, in terms of the cross-Wigner transform and the cross-ambiguity function.

The usefulness of Weyl--Heisenberg frames comes from the fact that they serve as \textquotedblleft generalized bases\textquotedblright\ in the Hilbert space $L^{2}({\mathbb{R}}^{n})$. The Weyl--Heisenberg expansion of an element $\psi$ in the Hilbert space $\Lt$ with respect to the window $\phi$ is given by
\begin{equation}\label{eq_Gabor_expansion}
	\psi = \sum_{z_\l \in \L} c_{z_\l} \, \widehat{T}(z_\l) \phi.
\end{equation}
Due to the possible over-completeness of the system (see Section \ref{sec_frame_set}), this expansion is not always unique and in general it is difficult to determine the coefficients $c_{z_\l}$.
One possibility to do so, is to introduce the canonical dual window to $\phi$, $\phi^\circ = \widehat{A}_{{\mathcal{G}}}^{-1} \phi$ (see equation \eqref{eq_frame_dual} above). For every Weyl--Heisenberg expansion of type \eqref{eq_Gabor_expansion} we have
\begin{equation}
	\sum_{z_\l \in \L} |c_{z_\l}|^2 \geq \sum_{z_\l \in \L} |(\psi|\widehat{T}(z_\l) \phi^\circ)|^2
\end{equation}
with equality only if
\begin{equation}
	c_{z_\l} = (\psi|\widehat{T}(z_\l) \phi^\circ).
\end{equation}

\subsection{Lattices, the Frame Set and an Approximation Formula}\label{sec_frame_set}

One of the most challenging questions in time-frequency analysis is to determine whether a Weyl--Heisenberg system already constitutes a frame. Only in this case we have a stable Weyl--Heisenberg expansion of every element in our Hilbert space with respect to the elements of the Weyl--Heisenberg system. This question seems to be far too general to be answered. At this point the \emph{frame set} enters the scene. For a fixed window $\phi$, the \textit{frame set} consists of all discrete point sets $\L \subset \R^{2n}$ which together with $\phi$ give a frame, that is
\begin{equation}\label{eq_frame_set}
	\mathfrak{F}(\phi) = \{ \L \subset \R^{2n} \, \colon \G(\phi,\L) \text{ is a frame}\}.
\end{equation}
For the 1-dimensional standard Gaussian function $\varphi^\hbar (x) = (\tfrac{1}{\pi \hbar})^{1/4} e^{-\tfrac{1}{2 \hbar} x^2}$ the system $\G(\varphi^\hbar,\L)$ is a frame for $\Lt[]$ whenever the lower Beurling density of $\L$ is greater than $(2 \pi \hbar)^{-1}$, i.e., if
\begin{equation}\label{eq_density}
	\delta_* (\L) = \liminf_{r \to \infty} \frac{1}{r^2} \min_{(x,p) \in \R^2} \# \{ z \in \L \, \colon z \in (x,p)+[0,r]^2 \} > (2 \pi \hbar)^{-1}.
\end{equation}
In this case, the necessary conditions imposed by the Balian-Low theorem and the density theorem are already sufficient as proved by Lyubarskii~\cite{Lyu92}, Seip~\cite{Sei92} and Seip and Wallst\'{e}n \cite{sewa92} (see also e.g.~\cite[chap.~8.4]{Gro01}).

Condition \eqref{eq_density} is one of the manifestations of the classical uncertainty principle and it implies that we cannot construct an orthonormal basis consisting of quantum displaced Gaussians. Allowing arbitrary point sets in phase space is too general for this work as we want to come up with explicit formulas. Therefore, we focus on lattices and for the rest of this work $\L$ will be a lattice, unless otherwise mentioned. We recall that a lattice is a discrete co-compact subgroup of $(\R^n, +)$ and that one can write a given lattice as
\begin{equation}
	\L = \mathcal{M} \, \Z^{2n}
\end{equation}
for some, not uniquely determined, invertible matrix $\mathcal{M} \in GL(2n,\R)$. The columns of $\mathcal{M}$ serve as a basis for $\L$ and its non-uniqueness results from the fact that we can choose from countably many bases. The volume and density of the lattice are given by
\begin{equation}
	\vol(\L) = |\det(\mathcal{M})| \quad \text{ and } \quad \delta(\L) = \frac{1}{\vol(\L)}
\end{equation}
respectively. For a lattice the Beurling density $\delta_*$ coincides with the density $\delta$ of the lattice. A lattice is called symplectic if the generating matrix is a multiple of a symplectic matrix, i.e., $\L = c S \, \Z^{2n}$ with $S \in Sp(n)$ and $c \neq 0$ (we may assume, without loss of generality, that $c > 0$ as $\L = -\L$). Note that for $n = 1$ every lattice is symplectic, whereas this is not true in higher dimensions. The adjoint lattice to $\L = \mathcal{M} \Z^{2n}$ is defined by
\begin{equation}	
	\L^\circ = J \mathcal{M}^{-T} \Z^{2n}.
\end{equation}	
If $\L$ is symplectic, we simply have $\L^\circ = \tfrac{1}{\vol(\L)^{1/n}} \L$. This follows from the definition of a symplectic matrix, $S J = J S^{-T}$, and $J \Z^{2n} = \Z^{2n}$. In this work we will exclusively consider symplectic lattices, although some of the results are valid for general index sets $\L$, which we will point out in those cases. 

There are some canonical subsets of $\mathfrak{F}(\phi)$ which we introduce now. We define the \textit{lattice frame set} of $\phi$ by
\begin{equation}
	\mathfrak{F}_{\L}(\phi) = \{\L \subset \R^{2n}, \, \L \text{ lattice } \colon \G(\phi,\L) \text{ is a frame}\}.
\end{equation}
A special class of lattices are the so-called separable lattices, which are of the form $\L_{\alpha, \beta} = \alpha \Z^n \times \beta \Z^n$. A generating matrix is given by
\begin{equation}
	S = \left(
	\begin{array}{cc}
		\alpha I & 0\\
		0 & \beta I
	\end{array}
	\right),
\end{equation}
and the density of the lattice is $\delta(\L_{\alpha, \beta}) = (\alpha \beta)^{-n}$. The \textit{reduced frame set} of $\phi$ is defined as
\begin{equation}
	\mathfrak{F}_{(\alpha,\beta)}(\phi) = \{ (\alpha, \beta) \in \R_+^2 \, \colon \G(\phi, \alpha \Z^n \times \beta \Z^n)  \text{ is a frame}\}.
\end{equation}
Clearly, if $(\alpha, \beta) \in \mathfrak{F}_{(\alpha,\beta)}(\phi)$ then $\alpha \Z^n \times \beta \Z^n \in \mathfrak{F}_{\L}(\phi)$. Hitherto, the only windows for which a complete characterization of the general and the lattice frame sets are known are generalized 1-dimensional Gaussians. There are some classes of functions for which the reduced frame set is fully known, we refer to the surveys by Gröchenig \cite{Gro14} and Heil \cite{Hei07} for more details. Note that for $n=1$ we have $Sp(1) = SL(2,\R)$ and that 1-dimensional generalized Gaussians can be written as the composition of a metaplectic operator and the standard Gaussian.


The metaplectic group $\operatorname*{Mp}(n)$ is a unitary representation of the double cover $\operatorname*{Sp}_{2}(n)$ of the symplectic group $\operatorname*{Sp}(n)$. The simplest (but not necessarily the most useful) way of describing $\operatorname*{Mp}(n)$ is to use its elementary generators $\widehat{J}$, $\widehat{V}_{-P}$, and $\widehat{M}_{L,m}$; denoting by $\pi^{\operatorname*{Mp}}$ the covering projection $\operatorname*{Mp}(n)\longrightarrow\operatorname*{Sp}(n)$ these operators and their projections are given by
\begin{align}
	\widehat{J}\psi(x) &  =e^{-in\pi/4} \mathcal{F} \psi(x),
	& & \pi^{\operatorname*{Mp}}(\widehat{J})=J\label{mp1}\\
	\widehat{V}_{-P}\psi(x) &  =e^{\frac{i}{2\hbar}Px^{2}}\psi(x),
	& & \pi^{\operatorname*{Mp}}(\widehat{V}_{-P})=V_{-P}\label{mp2}\\
	\widehat{M}_{L,m}\psi(x) &  =i^{m}\sqrt{|\det L|}\psi(Lx),
	& &	\pi^{\operatorname*{Mp}}(\widehat{M}_{L,m})=M_{L,m}.\label{mp3}
\end{align}
Here $\mathcal{F}$ is the unitary $\hbar$-Fourier transform and $V_{-P}$ ($P=P^{T}$), $M_{L,m}$ ($\det L\neq0$) are the symplectic generator matrices
\begin{equation}
	V_{-P} =
	\begin{pmatrix}
		I & 0\\
		P & I
	\end{pmatrix}
	\text{, \ }
	M_{L,m} =
	\begin{pmatrix}
		L^{-1} & 0\\
		0 & L^{T}
	\end{pmatrix}.
\end{equation}
The index $m$ in $\widehat{M}_{L,m}$ is the Maslov index, an integer corresponding to a choice of $\arg\det L$: $m$ is even if $\det L>0$ and odd if $\det L<0$.

Using these generators it is a simple exercise to show that the Heisenberg operators satisfy the symplectic covariance relations
\begin{equation}
	\widehat{T}(Sz_{0})=\widehat{S}\widehat{T}(z_{0})\widehat{S}^{-1}
\end{equation}
for every $\widehat{S}\in\operatorname*{Mp}(n)$, $S=\pi^{\operatorname*{Mp}}(\widehat{S})$.


It is always possible to construct frames with non-separable lattices from a given frame with separable lattice using the property of symplectic/metaplectic covariance as one of us has shown in \cite{ACHA} (see also \cite{Faulhuber_Invariance_2016}):

\begin{proposition}
	Let $\phi \in L^{2}(\mathbb{R}^{n})$. A Weyl--Heisenberg system $\mathcal{G}(\phi,\Lambda)$ is a frame if and only if $\mathcal{G}(\widehat{S}\phi,S\Lambda)$ is a frame; when this is the case both frames have the same bounds. In particular, $\mathcal{G}(\phi,\Lambda)$ is a tight frame if and only if $\mathcal{G}(\widehat{S}\phi, S \Lambda)$ is a tight frame.
\end{proposition}

From \cite{AB2,goetz} we recall the following generalization of the results originating from the work of Lyubarskii~\cite{Lyu92}, Seip~\cite{Sei92}, and Seip and Wallst\'{e}n~\cite{sewa92}:

\begin{proposition}
	For the multi-indices $\alpha=(\alpha_{1},\dots,\alpha_{n}\mathbb{)}, \beta=(\beta_{1}, \dots ,\beta_{n}\mathbb{)} \in \mathbb{Z}^n$ let $\Lambda_{\alpha, \beta} = \alpha{\mathbb{Z}}^{n}\times\beta{\mathbb{Z}}^{n}$. Let $\varphi^{\hbar}_j(x_j) = (\pi\hbar)^{-1/4}e^{-\tfrac{1}{2\hbar} x_j^{2}}$ be the 1-dimensional standard Gaussian and let $\varphi^{\hbar}(x)=(\pi\hbar)^{-n/4}e^{-\tfrac{1}{2\hbar} x^{2}} = \prod_{j=1}^n \varphi^{\hbar}_j(x_j)$ be the $n$-dimensional standard Gaussian. Then the following are equivalent:
		\begin{itemize}
		\item[(i)] $\mathcal{G}(\varphi^{\hbar},\Lambda_{\alpha\beta})$ is a frame.
		\item[(ii)] $\mathcal{G}(\varphi^{\hbar}_j,\Lambda_{\alpha_j \beta_j})$ is a frame for all $j = 1, \ldots, n$.
		\item[(iii)] $\alpha_{j}\beta_{j} < 2\pi\hbar$ for all $j = 1, \ldots, n$.
	\end{itemize}	
\end{proposition}

Combining these two results we get the following statement, whose proof can be found in de Gosson~\cite{ACHA}:

\begin{corollary}
	Let $\widehat{S}\in Mp(n)$ have projection $S\in Sp(n)$. The Weyl--Heisenberg system $\mathcal{G}(\widehat{S}\varphi^{\hbar},S\Lambda_{\alpha\beta})$ is a frame if and only if $\alpha_{j}\beta_{j}<2\pi\hbar$ for $1\leq j\leq n$. In this case the frame bounds of $\mathcal{G}(\widehat{S}\varphi^{\hbar},S\Lambda_{\alpha\beta})$ are the same as those of $\mathcal{G}(\varphi^\hbar,\Lambda_{\alpha\beta})$.
\end{corollary}
Note that for $n>1$, $Sp(n)$ is a proper subgroup of $SL(2n,\R)$ and, hence, the above results do not necessarily carry over to arbitrary lattices.

It was already observed by Folland \cite[Chap.~4]{Folland} that for $n = 1$ the only family of functions whose Wigner transforms are rotation invariant are the Hermite functions (which include the standard Gaussian). This is usually referred to as the ``rotational invariance" of Hermite functions; heuristically it gives an explanation for the fact that all systems ${\mathcal{G}}(\varphi^{\hbar},R\L)$ with $R \in SO(2n, \R)$ give a frame whenever $\delta_* > (2\pi\hbar)^{-1}$ (in this case $\L$ need not to be a lattice). Moreover, the frame bounds are the same regardless of $R$. See Faulhuber \cite{Faulhuber_Invariance_2016} and de Gosson \cite{JGP17} for an extension of this result to arbitrary Gaussian and Hermitian frames. This has led to the following conjecture in one of the author's doctoral thesis \cite{Faulhuber_PhD}:

\begin{conjecture} \label{ConjRotInv}
	For $\phi \in L^2(\R)$ the following are equivalent:
	\begin{enumerate}[(i)]
		\item $\phi$ is a Hermite function.
		\item $W\phi$ is rotation-invariant.
		\item The frames $\G(\phi,R\L)$ possess the same frame bounds for all $R \in SO(2, \R)$.
	\end{enumerate}
\end{conjecture}

For $n > 1$, Conjecture~\ref{ConjRotInv} in combination with Folland's result suggests the somewhat restricted conjecture.

\begin{conjecture} \label{ConjRotInv2}
	Let $n > 1$. Then for $\phi \in L^2(\R^n)$ the following are equivalent:
	\begin{enumerate}[(i)]
		\item $\phi$ is a Gaussian.
		\item $W\phi$ is rotation-invariant.
		\item The frames $\G(\phi,R\L)$ possess the same frame bounds for all $R \in O(2n, \R)$. 
	\end{enumerate}
\end{conjecture}
Let us note that for both conjectures the relations $(i) \Leftrightarrow (ii) \Rightarrow (iii)$ are well-known (see, e.g., \cite{Folland}) in contrast to $(iii) \Rightarrow (i)$.

We will discuss some properties of the Weyl--Heisenberg frame operator now.
Regarding the approximation we have the following result (for more details see \cite{FeiZim98}). For $\phi \in \Lt$ with $\norm{\phi} = 1$ we have
\begin{equation}\label{eq_approx_rate}
	\norm{Id - \vol(\L) \hspace{2pt} \widehat{A}_{\mathcal{G}}}_{op} \leq \sum_{z_\l^\circ \in \L^\circ \backslash \{0\}} |(\phi \mid \widehat{T}(z_\l^\circ) \phi)|.
\end{equation}
This result follows from Janssen's representation of the Weyl--Heisenberg frame operator $\widehat{A}_{\mathcal{G}}$ \cite{Janssen_Duality_1994}. The idea behind it is the following. If the lattice $\L$ is quite dense, then the adjoint lattice $\L^\circ$ is rather sparse and since $\bigl( \phi \mid \widehat{T}(z_\l^\circ) \phi \bigr) \in C_0(\R^{2d})$, the sum on the right-hand side of $\eqref{eq_approx_rate}$ tends to zero as $\vol(\L) \to 0$. The speed of convergence depends on the concrete function as well as on the lattice (see the works of Faulhuber \cite{Faulhuber_Hexagon_2017} and Faulhuber and Steinerberger \cite{Faulhuber_Steinerberger_2017} for a study on optimal lattices for Gaussian Weyl--Heisenberg frames). For Gaussians the convergence is of exponential order with respect to the density $\delta(\L) = \tfrac{1}{\vol(\L)}$ (see Sections~\ref{sec_Gauss_n} and \ref{sec_Gauss_1} for exact results).
It follows that the frame operator satisfies
\begin{equation}
	\lim_{\vol{\L} \to 0} \vol(\L) \hspace{2pt} \widehat{A}_{\mathcal{G}} = Id.
\end{equation}
So, as the density increases, the frame operator converges to the identity operator in the operator norm. The analogous statement is of course true for the inverse frame operator, meaning that
\begin{equation}
	\lim_{\vol(\L) \to 0} \tfrac{1}{\vol(\L)} \hspace{2pt} \widehat{A}_{\mathcal{G}}^{-1} = Id
\end{equation}
in the operator norm. This implies
\begin{equation}
	\lim_{\vol(\L) \to 0} \tfrac{1}{\vol(\L)} \hspace{2pt} \phi^\circ = \phi
\end{equation}
in the Hilbert space norm. At this point, it seems appropriate to introduce Feichtinger's algebra $\mathcal{S}_0 (\R^n)$ which (densely) contains the Schwartz space $\mathcal{S}(\R^n)$. Feichtinger's algebra, introduced by Feichtinger in the 1980s \cite{Fei81}, is easily characterized as follows;
\begin{equation}
	\phi \in \mathcal{S}_0(\R^n) \quad \Longleftrightarrow \quad W\phi \in L^1(\R^{2n}).
\end{equation}
It is a Banach space invariant under the Fourier transform and the action of the Heisenberg operators $\widehat{T}$. For more details we refer to the study by Jakobsen \cite{Jak18}. If $\phi \in \mathcal{S}_0(\R^n)$, then $\tfrac{1}{\vol(\L)}\phi^\circ$ also converges uniformly to the window $\phi$ \cite{FeiZim98}. Hence, for large density (small volume) of the lattice, we have the approximate Weyl--Heisenberg expansion
\begin{equation}\label{eq_approx_expansion}
	\psi(x) \approx \vol(\L) \sum_{z_\l \in \L} (\psi \, | \, \widehat{T}(z_\l) \phi ) \; \widehat{T}(z_\l) \phi(x).
\end{equation}

\subsection{Two Reformulations of the Frame Condition}\label{sec_reformulation}

We will see that Weyl--Heisenberg frames can be expressed both in terms of the cross-Wigner transform and the cross-ambiguity function.

Recall that the cross-Wigner transform of a pair of square-integrable functions $(\psi,\phi)$ is
\begin{equation}\label{w}
	W(\psi,\phi)(x,p)=\left(  \tfrac{1}{2\pi\hbar}\right)^{n} \int_{\R^n} e^{-\frac{i}{\hbar}py}\psi(x+\tfrac{1}{2}y)\overline{\phi(x-\tfrac{1}{2}y)} \, d^{n}y;
\end{equation}
the cross-ambiguity function of $(\psi,\phi)$ is in turn given by
\begin{equation}\label{a}
	A(\psi,\phi)(x,p)=\left(  \tfrac{1}{2\pi\hbar}\right)^{n} \int_{\R^n} e^{-\tfrac{i}{\hbar}py}\psi(y+\tfrac{1}{2}x)\overline{\phi(y-\tfrac{1}{2}x)} \, d^{n}y.
\end{equation}
For the (auto) Wigner transform and the ambiguity function we write
\begin{equation}
	 W\phi = W(\phi,\phi) \, \text{ and } \, A\phi = A(\phi,\phi),
\end{equation}
respectively. It was observed already by Klauder that these functions are (symplectic) Fourier transforms of each other, that is
\begin{equation}\label{rel1}
	W(\psi,\phi)=\F_{\sigma}A(\psi,\phi) \quad \text{ and } \quad A(\psi,\phi)=\F_{\sigma} W(\psi,\phi),
\end{equation}
for $\F_{\sigma}\psi(z) = \F \psi(Jz)$. We also have the algebraic relation
\begin{equation}
A(\psi,\phi)(z)=2^{-n}W(\psi,\phi^{\vee})(\tfrac{1}{2}z) \label{rel2},
\end{equation}
where $\phi^{\vee}(x) = \phi(-x)$ (see for instance \cite{Folland,Birkbis}). The cross-Wigner transform satisfies the Moyal identity
\begin{equation}\label{Moyal1}
	\bigl( \bigl( W(\psi,\phi)|W(\psi^{\prime},\phi^{\prime}) \bigr) \bigr) = \left(  \tfrac{1}{2\pi\hbar}\right)  ^{n}(\psi|\psi^{\prime})\overline{(\phi|\phi^{\prime})}
\end{equation}
for all $\psi,\phi,\psi^{\prime},\phi^{\prime}$ in $L^{2}({\mathbb{R}}^{n})$; using Plancherel's formula together with \eqref{rel2}, we also have
\begin{equation}\label{Moyalamb}
	\bigl( \bigl( A(\psi,\phi)|A(\psi^{\prime},\phi^{\prime}) \bigr) \bigr) = \left(  \tfrac{1}{2\pi\hbar}\right)  ^{n}(\psi|\psi^{\prime})\overline{(\phi|\phi^{\prime})}.
\end{equation}

To reformulate the frame conditions, we will need the following lemma:

\begin{lemma}
	For all $\psi,\phi \in L^{2}({\mathbb{R}}^{n})$
	we have
	\begin{equation}\label{atz}
		A(\psi,\phi)(z)=\left(  \tfrac{1}{2\pi\hbar}\right)  ^{n}(\psi|\widehat{T}(z)\phi).
	\end{equation}
\end{lemma}

We omit the proof as it is a trivial consequence of \eqref{heiwe} (see \cite{Birk,Birkbis}).

\begin{proposition}
	The system ${\mathcal{G}}(\phi,\Lambda)$ is a \textit{Weyl--Heisenberg frame with bounds }$a,b$ if and only if either of the following two equivalent conditions hold for all $\psi\in L^{2}({\mathbb{R}}^{n})$:
	\begin{align}
		(2\pi\hbar)^{2n}a||\psi||^{2}  &  \leq\sum_{z_{\lambda}\in\Lambda}|A(\psi,\phi)(z_{\lambda})|^{2}\leq(2\pi\hbar)^{2n}b||\psi||^{2}\label{k1}\\
		(4\pi\hbar)^{2n}a||\psi||^{2}  &  \leq\sum_{z_{\lambda}\in\frac{1}{2}\Lambda}|W(\phi,\psi)(z_{\lambda})|^{2}\leq(4\pi\hbar)^{2n}b||\psi||^{2}\label{k2}
	\end{align}
\end{proposition}

\begin{proof}
	Condition \eqref{k1} immediately follows from \eqref{atz} and, by \eqref{rel2}, it is equivalent to \eqref{k2}.
\end{proof}

\subsection{Extension of a Frame to Phase Space}

The cross-Wigner transform satisfies the translational property
\begin{equation}\label{ttfn}
	W(\widehat{T}(z_{0})\psi,\widehat{T}(z_{1})\phi)(z)=e^{-\frac{i}{\hbar}[\sigma(z,z_{0}-z_{1})+\frac{1}{2}\sigma(z_{0},z_{1})]} \hspace{2pt} W(\psi,\phi)(z-\tfrac{1}{2}(z_{0}+z_{1}))\nonumber
\end{equation}
for all $\psi,\phi\in\mathcal{S}^{\prime}(\mathbb{R}^{n})$ (see \cite{Folland,Birkbis,gowig}). In particular, taking $z_{1}=0$, we have
\begin{equation}
	W(\widehat{T}(z_{0})\psi,\phi)(z)=e^{-\frac{i}{\hbar}\sigma(z,z_{0})}W(\psi,\phi)(z-\tfrac{1}{2}z_{0}),
\end{equation}
which motivates the notation
\begin{equation}\label{interwig}
	\widetilde{T}(z_{0})W(\psi,\phi) = W(\widehat{T}(z_{0})\psi,\phi).
\end{equation}
This suggests to define, as in \cite{Birkbis}, the operators
\begin{align}\label{tildezo1}
	&\widetilde{T}(z_{0}): \mathcal{S}(\mathbb{R}^{2n}) \to \mathcal{S}(\mathbb{R}^{2n}), \\
	&\widetilde{T}(z_{0})\Psi(z) = e^{-\frac{i}{\hbar}\sigma(z,z_{0})}\Psi(z-\tfrac{1}{2}z_{0}).
\end{align}
These operators extend to unitary operators on $L^{2}(\mathbb{R}^{2n})$ and, a fortiori, to continuous automorphisms of $\mathcal{S}^{\prime}(\mathbb{R}^{2n})$. We have
\begin{equation}
	\widetilde{T}(z_{0})^{\ast}=\widetilde{T}(z_{0})^{-1}=\widetilde{T}(-z_{0}),
\end{equation}
and it is easily checked that the operators $\widetilde{T}(z_{0})$ satisfy the same commutation and addition properties as the Heisenberg operators, that is
\begin{align}
	\widetilde{T}(z_{0})\widetilde{T}(z_{1})  &  =e^{\tfrac{i}{\hslash}\sigma(z_{0},z_{1})}\widetilde{T}(z_{1})\widetilde{T}(z_{0})\label{HW3}\\
	\widetilde{T}(z_{0}+z_{1})  &  =e^{-\tfrac{i}{2\hslash}\sigma(z_{0},z_{1})}\widetilde{T}(z_{0})\widetilde{T}(z_{1}).\label{HW4}
\end{align}
One of us has studied these operators in relation with the extension of Weyl operators to phase space \cite{Birkbis}. This extension works as follows: let $a\in\mathcal{S}^{\prime}(\mathbb{R}^{2n})$ be viewed as a symbol; the corresponding Weyl operator
\begin{equation}
	\operatorname*{Op}\nolimits^{\mathrm{W}}(a):\mathcal{S}(\mathbb{R}^{n})\longrightarrow\mathcal{S}^{\prime}(\mathbb{R}^{n})
\end{equation}
is defined by
\begin{equation}\label{defweyl}
	\operatorname*{Op}\nolimits^{\mathrm{W}}(a) \psi(x)=\left(  \tfrac{1}{2\pi\hbar}\right)^{n} \int_{\mathbb{R}^{2n}}a_{\sigma}(z_{0})\widehat{T}(z_{0})\psi(x)d^{2n}z_{0}
\end{equation}
where $a_{\sigma}$ is the symplectic Fourier transform of $a$, formally given by
\begin{equation}
	a_{\sigma}(z)=\left(  \tfrac{1}{2\pi\hbar}\right)^{n} \int_{\mathbb{R}^{2n}}e^{-\frac{i}{\hbar}\sigma(z,z^{\prime})}a(z^{\prime})d^{2n}z^{\prime}.
\end{equation}
For example, we have $\widehat{\rho} = \operatorname*{Op}\nolimits^{\mathrm{W}}(\rho)$ for the density operator $\widehat{\rho}$ and its Wigner transform $\rho$ mentioned in the introduction.

One then defines the ``Bopp operator" $\operatorname*{Op}\nolimits^{\mathrm{B}}(a)$ (see \cite{goetal,Birkbis,gobopp2}) by replacing $\widehat{T}(z_{0})$ in \eqref{defweyl} by $\widetilde{T}(z_{0})$: for $\Psi\in\mathcal{S}(\mathbb{R}^{2n})$ we thus have
\begin{equation}\label{defatilde}
	\operatorname*{Op}\nolimits^{\mathrm{B}}(a) \Psi(z)=\left(  \tfrac{1}{2\pi\hbar}\right)^{n} \int_{\mathbb{R}^{2n}}a_{\sigma}(z_{0})\widetilde{T}(z_{0})\Psi(z)d^{2n}z_{0}.
\end{equation}
This quantization associates to the symbols $x$ and $p$ the operators
	\begin{align*}
		\tilde{x} = x + \frac{i \hbar}{2} \partial p \hspace{5pt} \mbox{and} \hspace{5pt} \tilde{p} = p - \frac{i \hbar}{2} \partial x,
	\end{align*}
respectively.

\subsection{Definition of a Phase Space Frame}

For fixed $\phi\in L^{2}({\mathbb{R}}^{n})$ we define the operator
\begin{align}
	U_{\phi}:L^{2}({\mathbb{R}}^{n}) &\to L^{2}({\mathbb{R}}^{2n}), \\
	\psi &\mapsto \Psi = (2\pi\hbar)^{n/2}W(\psi,\phi). \label{wpt}
\end{align}
By Moyal's formula we have $((U_{\phi}\psi|U_{\phi}\psi^{\prime}))=(\psi|\psi^{\prime})$, hence $U_{\phi}$ is an isometry of $L^{2}({\mathbb{R}}^{n})$ onto proper subspace of $L^{2}(\mathbb{R}^{2n})$, which we denote by $\mathcal{H}_{\phi}$. Observe that $U_{\phi}^{\ast}U_{\phi}$ is the identity operator on $L^{2}({\mathbb{R}}^{n})$ and that $\Pi_{\phi}=U_{\phi}U_{\phi}^{\ast}$ is the orthogonal projection of $L^{2}(\mathbb{R}^{2n})$ onto $\mathcal{H}_{\phi}$ (we have $\Pi_{\phi}^{\ast}=\Pi_{\phi}$, $\Pi_{\phi}^{2}=\Pi_{\phi}$ and the range of $U_{\phi}^{\ast}$ is $L^{2}({\mathbb{R}}^{n})$).

Let ${\mathcal{G}}(\Lambda,\phi)$ be a Weyl--Heisenberg frame. We associate to ${\mathcal{G}}(\Lambda,\phi)$ the phase space frame operator
\begin{equation}\label{apfi}
	\widetilde{A}_{{\mathcal{G}}}\Psi(x) = \sum_{z_{\lambda}\in\Lambda}((\Psi|\widetilde{T}(z_{\lambda})\Phi)) \, \widetilde{T}(z_{\lambda})\Phi,
\end{equation}
where $\Phi=U_{\phi}\phi=(2\pi\hbar)^{n/2}W\phi$.

\begin{proposition}
	Let $\widehat{A}_{{\mathcal{G}}}$ be the frame operator of ${\mathcal{G}}(\Lambda,\phi)$ and $\widetilde{A}_{{\mathcal{G}}}$ the associated phase space frame operator. Then we have the intertwining relation
	\begin{equation}\label{affi}
		\widetilde{A}_{{\mathcal{G}}}\Psi=U_{\phi}(\widehat{A}_{\G}\psi)
	\end{equation}
	for all $\Psi\in\mathcal{H}_{\phi}$.
\end{proposition}

\begin{proof}
	Since every $\Psi\in\mathcal{H}_{\phi}$ can be written as
	\begin{equation}
		\Psi=U_{\phi}\psi=(2\pi\hbar)^{n/2}W(\psi,\phi)
	\end{equation}
	for some uniquely determined $\psi\in L^{2}({\mathbb{R}}^{n})$, it suffices to show
	\begin{equation}\label{aga}
		\widetilde{A}_{{\mathcal{G}}}W(\psi,\phi)=W(\widehat{A}_{{\mathcal{G}}}\psi,\phi).
	\end{equation}
	for all $\psi \in L^2(\R^n)$. Using \eqref{interwig}, we have
	\begin{align}
		\widetilde{A}_{{\mathcal{G}}}W(\psi,\phi)  &  =(2\pi\hbar)^{n}\sum_{z_{\lambda} \in \Lambda} ((W(\psi,\phi)|\widetilde{T}(z_{\lambda})W\phi)) \hspace{2pt} \widetilde{T}(z_{\lambda}) W\phi\\
		&  =(2\pi\hbar)^{n}\sum_{z_{\lambda} \in \Lambda} ((W(\psi,\phi)|W(\widehat{T}(z_{\lambda})\phi,\phi)) \hspace{2pt} W(\widehat{T}(z_{\lambda})\phi,\phi),
	\end{align}
	on the one hand. On the other hand, by Moyal's identity \eqref{Moyal1}, we have
	\begin{equation}
		((W(\psi,\phi)|W(\widehat{T}(z_{\lambda})\phi,\phi)))=\left(  \tfrac{1}{2\pi\hbar}\right)  ^{n}(\psi|\widehat{T}(z_{\lambda})\phi)
	\end{equation}
	since $\phi$ is normalized, and hence
	\begin{align}
		\widetilde{A}_{{\mathcal{G}}}W(\psi,\phi) & = \sum_{z_{\lambda} \in \Lambda} (\psi|\widehat{T}(z_{\lambda})\phi) \hspace{2pt} W \left(\widehat{T}(z_{\lambda})\phi, \, \phi \right)\\
		& = \sum_{z_{\lambda} \in \Lambda} W\left((\psi|\widehat{T}(z_{\lambda})\phi) \hspace{2pt} \widehat{T}(z_{\lambda})\phi, \, \phi \right).
	\end{align}
	This is precisely \eqref{aga}.
\end{proof}

\begin{corollary}
	Let ${\mathcal{G}}(\phi,\Lambda)$ be a Weyl--Heisenberg frame and set $\Phi=U_{\phi}\phi$. Then the system $\widetilde{\G}(\Phi,\Lambda)=\{\widetilde{T}(z_{\lambda})\Phi:z_{\lambda}\in\Lambda\}$ is a frame on the Hilbert space $\mathcal{H}_{\phi}$ with frame operator $\widetilde{A}_{\G}$ and we have
	\begin{equation}
		a|||\Psi|||^{2}\leq\sum_{z_{\lambda}\in\Lambda}|((\Psi|\widetilde{T}(z_{\lambda})\Phi))|^{2}\leq b|||\Psi|||^{2} \label{frametilda}
	\end{equation}
	for all $\Psi\in\mathcal{H}_{\phi}$ where $a$ and $b$ are the same bounds as for the system $\mathcal{G}(\phi, \Lambda)$.
\end{corollary}

\begin{proof}
It is an immediate consequence of \eqref{eq_frame} by employing the isometry $U_\phi$.
\end{proof}

\section{Gaussian Mixed States on \texorpdfstring{$\R^n$}{R*n}}\label{sec_Gauss_n}

In this section we derive all results stated in Theorems~\ref{thm_main} and \ref{thm_main2}. We start with a quick characterization of Gaussian mixed states, which is followed by a description and explicit formulas of the cross-Wigner transform and the cross-ambiguity function of pairs of Gaussians. We conclude the section with explicit phase space frame expansions in terms of standard phase space Gaussians.

\subsection{Characterization}

For the moment, we go back to using the notation of density operators from the introduction. Let $\rho$ be a Gaussian of the type
\begin{equation}\label{rhobis}
	\rho(z) = \sqrt{\det\Sigma^{-1}}e^{-\pi \, \Sigma^{-1}z^{2}},
\end{equation}
where the covariance matrix $\Sigma$ is a real positive definite symmetric matrix. The function $\rho$ is normalized such that
\begin{equation}
	\int_{{\mathbb{R}}^{2n}}\rho(z)d^{2n}z=1
\end{equation}
so that it can be viewed as a probability density. As briefly discussed in the introduction, $\rho$ is the Wigner distribution of a density operator $\widehat{\rho}$ if and only if $\Sigma$ satisfies the condition\footnote{Observe that $\Sigma+\frac{i\hbar}{2}J$ is a self-adjoint complex matrix since $J^{T}=-J=J^{-1}$; it follows that its eigenvalues are real, and the condition above is equivalent to saying that these eigenvalues are all non-negative.}
\begin{equation}
	\Sigma+\frac{i\hbar}{2}J\geq0.
\end{equation}

\subsection{The Cross-Wigner and Cross-Ambiguity Transform of Pairs of Gaussians}

Let $\varphi^\hbar_{M}$ be the centered Gaussian
\begin{equation}\label{eq_gauss}
	\varphi^\hbar_{M}(x) = \left(  \tfrac{1}{\pi\hbar}\right)  ^{n/4}(\det X)^{1/4} e^{-\tfrac{1}{2\hbar}Mx^{2}},
\end{equation}
where $M=X+iY$ with positive definite, symmetric $X$ and symmetric $Y$. The coefficient in front of the exponential is chosen such that $\varphi^\hbar_M$ is $L^2$-normalized. If $M = I$, we simply write $\varphi^{\hbar}$ instead of $\varphi^\hbar_I$.

In \cite{Birkbis} one of us has shown the following result.

\begin{proposition}\label{pro_cross_Wigner}
	Let $(\varphi^\hbar_M,\varphi^\hbar_{M'})$ be a pair of Gaussians of the type \eqref{eq_gauss}. Their cross-Wigner transform is given by
	\begin{equation}
		W(\varphi^\hbar_M,\varphi^\hbar_{M'})(z)=\left(  \tfrac{1}{\pi\hbar}\right)^{n}C_{M,M'}e^{-\frac{1}{\hbar}Fz^{2}}\label{wififi},
	\end{equation}
	where $C_{M,M'}$ is the complex constant
	\begin{equation}
		C_{M,M'}=(\det XX')^{1/4}\det\left[  \tfrac{1}{2}(M+\overline{M'})\right] ^{-1/2}\label{cmm}
	\end{equation}
	and $F$ is the symmetric complex matrix
	\begin{equation}\label{wigf}
		F=
		\begin{pmatrix}
			2\overline{M'}(M+\overline{M'})^{-1}M & -i(M-\overline{M'})(M+\overline{M'})^{-1}\\
			-i(M+\overline{M'})^{-1}(M-\overline{M'}) & 2(M+\overline{M'})^{-1}
		\end{pmatrix}.
	\end{equation}
	If $M=M'=X+iY$ with positive definite, symmetric $X$ and symmetric $Y$, we write $G = F$ and recover the well-known formula
	\begin{align}\label{eq_auto_Wigner_Gaussian}
		W \varphi^\hbar_M (z) = \bigl( \tfrac{1}{\pi \hbar} \bigr)^n \hspace{2pt} e^{-\frac{1}{\hbar} G z^2}.
	\end{align}
	The matrix $G$ is then a real Gram matrix which can be factored as
	\begin{equation}\label{eq_splitting_Gram}
		G = S^T S,
	\end{equation}
	where the symplectic matrix $S$ is given by
	\begin{align}\label{eq_splitting_Gram_S}
		S =
		\begin{pmatrix}
		X^{1/2} & 0\\
		X^{-1/2} Y & X^{-1/2}
		\end{pmatrix}.
	\end{align}
\end{proposition}

\subsection{Gaussian Phase Space Frame Expansions}

In the following we will give a quantitative estimate for the approximation error \eqref{eq_approx_rate} for arbitrary symplectic lattices in $\mathbb{R}^n$ as the density tends to infinity. Recall that for the symplectic lattice $\L = \delta^{-1} S \Z^{2n}$, $S \in Sp(n)$, of density $\delta^{2n} = \vol(\L)^{-1}$, the adjoint lattice is simply given by $\L^\circ = \delta S \Z^{2n} = \vol(\L)^{-1/n} \L$. Given the Weyl--Heisenberg frame $\G(\varphi^\hbar,\L)$, Janssen's representation \cite{Janssen_Duality_1994} of the associated frame operator is given by
\begin{equation} \label{Janssen}
	\hat{A}_\mathcal{G}
	= \frac{(2 \pi \hbar)^n}{\vol(\L)} \hspace{2pt} \sum_{z_{\l^\circ} \in \L^\circ} A\varphi^\hbar(z_{\l^\circ}) \hspace{2pt} \widehat{T}(z_{\l^\circ})
	= \frac{(2 \pi \hbar)^n}{\vol(\L)} \hspace{2pt} \sum_{k, l \in \Z^n} A\varphi^\hbar \bigl(S (\delta k, \delta l) \bigr).
\end{equation}
The value of the ambiguity function evaluated at a point on the adjoint lattice is given by
\begin{align}
	A\varphi^\hbar \bigl( \delta S z_\l \bigr) = \left(2 \pi \hbar \right)^{-n} \hspace{2pt} e^{-\frac{\delta^2}{4 \hbar} G z_\l^2}, \quad z_\l \in \L,
\end{align}
where $G = S^T S$ is positive definite with
\begin{align}
	S =
	\begin{pmatrix}
		L & 0\\
		0 & L^{-T}
	\end{pmatrix}
		\begin{pmatrix}
		1 & 0\\
		P & 1
	\end{pmatrix}
	=
	\begin{pmatrix}
		L & 0\\
		L^{-T} P & L^{-T}
	\end{pmatrix} \in Sp(n),
\end{align}
for some invertible matrix $L$ and real, symmetric matrix $P = P^T$. For our purpose, this particular splitting is no restriction (recall \eqref{eq_auto_Wigner_Gaussian}, \eqref{eq_splitting_Gram}, \eqref{eq_splitting_Gram_S} as well as \cite[Prop.4.76.]{Folland}); this is due to the fact that for $S' = Q S$ with $Q \in SO(2n,\R) \cap Sp(n)$ we have $(S')^T S' = S^T S$. We obtain the following generalization of \cite[Prop.~3.1.]{Faulhuber_Hexagon_2017}, where the following result was proven for $n=1$.

\begin{proposition} \label{pro_error_estimate_WH}
	Let $\G(\varphi^\hbar, \L)$ be the Weyl--Heisenberg frame with $\L = \delta^{-1} Q S \L$, with $\delta > 0$, $Q, S$ and $\varphi^\hbar$ as above. Then for the Weyl--Heisenberg frame operator $\widehat{A}_\G$ we have
	\begin{align}
		|| Id - \delta^{-2n} \hspace{2pt} \hat{A}_\mathcal{G} ||_{op} = \mathcal{O}\Bigl( e^{-\frac{\delta^2}{4 \hbar} ( ||L||^2_{\R^{n \times n}} + ||L||^{-2}_{\R^{n \times n}})} \Bigr).
	\end{align}
\end{proposition}

\begin{proof}
	Let $\delta > 0$. Due to \eqref{Janssen}, we find
	\begin{align}
		\vol(\L) \hspace{2pt }\hat{A}_\mathcal{G} &= \sum_{k, l \in \Z^n} e^{-\frac{\delta^2}{4 \hbar} |S(k, l)|^2} \\
		&= \sum_{k, l \in \Z^n} e^{-\frac{\delta^2}{4 \hbar} L k^2} \hspace{2pt} e^{-\frac{\delta^2}{4 \hbar} |L^{-T}(Pk+ l)|^2} \\
		& \leq \sum_{k, l \in \Z^n} e^{-\frac{\delta^2}{4 \hbar} ||L||^2_{\R^{n \times n}} k^2} \hspace{2pt} e^{-\frac{\delta^2}{4 \hbar}||L||^2_{\R^{n \times n}} |Pk + l|^2} \\
		&= \sum_{k, m \in \Z^n} e^{-\frac{ \delta^2}{4 \hbar} ||L||^2_{\R^{n \times n}} k^2} \hspace{2pt} \frac{||L||^2_{\R^{n \times n}}}{\sqrt{2} \delta} \hspace{2pt} e^{-\frac{1}{\delta^2 \hbar} ||L||^2_{\R^{n \times n}} m^2} \hspace{2pt} e^{\frac{i}{\hbar} k P m}, \label{PoissonTerm}
	\end{align}
	where the last step is due to Poisson summation. As our summation goes over all $m \in \mathbb{Z}^n$, \eqref{PoissonTerm} equals
	\begin{align}
		\sum_{k, m \in \Z^n} e^{-\frac{ \delta^2}{4 \hbar} ||L||^2_{\R^{n \times n}} k^2} \hspace{2pt} \frac{||L||^2_{\R^{n \times n}}}{\sqrt{2} \delta} \hspace{2pt} e^{-\frac{1}{\delta^2 \hbar} ||L||^2_{\R^{n \times n}} m^2} \hspace{2pt} \cos\bigl( \tfrac{1}{\hbar} k P m \bigr).
	\end{align}
	Finally, since $\cos(2x) = 1 - 2 \sin(x)^2$, we obtain
	\begin{align}
		\vol(\L) \hspace{2pt }\hat{A}_\mathcal{G} &\leq \sum_{k, m \in \Z^n} e^{-\frac{\delta^2}{4 \hbar} ||L||^2_{\R^{n \times n}} k^2} \hspace{2pt} \frac{||L||^2_{\R^{n \times n}}}{\sqrt{2}  \delta} \hspace{2pt} e^{-\frac{1}{\delta^2 \hbar} ||L||^2_{\R^{n \times n}} m^2} \\
		&= \sum_{k, l \in \Z^n} e^{-\frac{\delta^2}{4 \hbar} \bigl( ||L||^2_{\R^{n \times n}}  k^2 + ||L||^{-2}_{\R^{n \times n}} l^2 \bigr)},
	\end{align}
	where we have applied Poisson summation once more.
\end{proof}

\begin{definition}\label{def_phase_space_Gaussian}
	For $\varphi^\hbar_M, \varphi^\hbar_{M'}$ as in \eqref{eq_gauss} we set
	\begin{align}
		\Phi^\hbar_F(z) = (2\pi\hbar)^{n/2} \hspace{2pt} W (\varphi^\hbar_M, \varphi^\hbar_{M'}) (z) = \bigl( \tfrac{2}{\pi \hbar} \bigr)^{n/2} C_{M,M'} \hspace{2pt} e^{-\frac{1}{\hbar} F z^2},
	\end{align}
	where $C_{M,M'}$ and $F$ are as in Proposition \ref{pro_cross_Wigner}. If $F = I$ we simply write $\Phi^\hbar$ instead of $\Phi^\hbar_I$.
\end{definition}

\begin{corollary}\label{cor_error_estimate_phase_space}
	Let $\L = \delta^{-1} Q S \Z^{2n}$. For the associated phase space frame $\G(\Phi^\hbar,\L)$ for $\mathcal{H}_{\varphi^\hbar} \subset \Lt[2n]$ associated to the Weyl--Heisenberg frame $\G(\varphi^\hbar, \L)$ we have the stable approximation
	\begin{align}
		|| Id - \delta^{-2n} \hspace{2pt} \widetilde{A}_\mathcal{G} ||_{op} = \mathcal{O}\Bigl( e^{-\frac{\delta^2}{4 \hbar} ( ||L||^2_{\R^{n \times n}} + ||L||^{-2}_{\R^{n \times n}})} \Bigr).
	\end{align}
\end{corollary}
\begin{proof}
	Since the according property of the Weyl--Heisenberg frame $\G(\Phi^\hbar,\L)$ holds due to Proposition~\ref{pro_error_estimate_WH}, the statement follows directly 	from an application of the isometry $U^\hbar_\varphi$ as $\Phi^\hbar = U_\varphi^\hbar \varphi^\hbar$.
\end{proof}

This proves Theorem~\ref{thm_main}.

\begin{proposition}\label{pro_coeff}
	Let $\Lambda \subset \mathbb{R}^{2n}$ be a discrete index set such that $\G(\Phi^\hbar,\L)$ is a phase space frame for $\mathcal{H}_{\varphi^\hbar}$. Then the coefficients of the expansion 
	\begin{equation}
		\widetilde{A}_\G \Phi_F^\hbar = \sum_{z_\l \in \L} c_{z_\l} \widetilde{T}(z_\l) \Phi^\hbar
	\end{equation}
	are given by
	\begin{align}
		c_{z_\l} = ((\Phi^\hbar_F, \widetilde{T}(z_\l) \Phi^\hbar))
		= \bigl( \tfrac{2}{\pi \hbar} \bigr)^{n/2} \hspace{2pt} \det(F+I)^{-1/2} \hspace{2pt} e^{-\frac{1}{4 \hbar} z_\l^2}
		\hspace{2pt} e^{-\tfrac{1}{4 \hbar} (F+I)^{-1} \bigl( (J -iI) z_\l \bigr)^2}.
	\end{align}
	Equivalently, by setting $H = (J -iI)^T (F+I)^{-1} \hspace{2pt} (J -iI)$, we have
	\begin{align}
		((\Phi^\hbar_F, \widetilde{T}(z_\l) \Phi^\hbar))
		= \bigl( \tfrac{2}{\pi \hbar} \bigr)^{n/2} \hspace{2pt} \hspace{2pt} \det(F+I)^{-1/2} \hspace{2pt} e^{-\tfrac{1}{4 \hbar} (I+H) z_\l^2}.
	\end{align}
\end{proposition}

\begin{proof}
	Let $z = (x,p) \in \R^{2n}$ and $z_\l =(x_\l,p_\l) \in \L$. Then we compute
	\begin{align}
		((\Phi^\hbar_F, \widetilde{T}(z_\l) \Phi^\hbar))
		&= \int_{\mathbb{R}^{2n}} \bigl( \tfrac{2}{\pi \hbar} \bigr)^{n/2} \hspace{2pt} e^{-\frac{1}{\hbar} F z^2}
			\bigl( \tfrac{2}{\pi \hbar} \bigr)^{n/2} e^{\frac{i}{\hbar}\sigma(z,z_l)} \hspace{2pt} e^{-\frac{1}{\hbar}(z-\tfrac{z_\l}{2})^2} \,d^{2n}z \\
		&= \bigl( \tfrac{2}{\pi \hbar} \bigr)^n \hspace{2pt} e^{-\frac{1}{4 \hbar}(x_\l^2 + p_\l^2)} \int_{\mathbb{R}^{2n}} e^{-\frac{1}{\hbar} (F+I) z^2}
			\hspace{2pt} e^{\frac{i}{\hbar}[x(p_\l - ix_\l) p (x_\l + ip_\l)]} \,d^n x \,d^n p \\
		&= \bigl( \tfrac{2}{\pi \hbar} \bigr)^n \hspace{2pt} e^{-\frac{1}{4 \hbar} z_\l^2} \int_{\mathbb{R}^{2n}} e^{-\frac{1}{2\hbar} 2(F+I) z^2}
			\hspace{2pt} e^{\frac{i}{\hbar}(J -i I) z_\l^2} \,d^{2n}z,
	\end{align}
	where in the last equality we have used $(J - iI)z_\l = (p_\l, x_\l) - i (x_\l, p_\l)$. The last integral can be viewed as the $2n$-dimensional Fourier transform of the Gaussian $(2\pi \hbar)^{n/2} e^{-\frac{1}{2 \hbar} 2(F+I) z^2}$ at $(iI-J)z_\l$ (see Appendix \ref{app_FT_Gauss}). We have
	\begin{equation}
		\int_{\mathbb{R}^{2n}} e^{-\frac{1}{2\hbar} 2(F+I) z^2} \hspace{2pt} e^{\frac{i}{\hbar}(J -i I) z_\l^2} \,d^{2n}z =
			\tfrac{(2 \pi \hbar)^{n/2}}{\det(2(F+I))^{1/2}} \hspace{2pt} e^{-\tfrac{1}{4\hbar} (F+I)^{-1} \left((iI-J)z_\l\right)^2}.
	\end{equation}
	Consequently,
	\begin{align}
		((\Phi^\hbar_F, \widetilde{T}(z_\l) \Phi^\hbar)) &= \bigl( \tfrac{2}{\pi \hbar} \bigr)^{n/2} \hspace{2pt} \det(F+I)^{-1/2} \hspace{2pt} e^{-\frac{1}{4 \hbar} z_\l^2}
		\hspace{2pt} e^{-\tfrac{1}{4 \hbar} (F+I)^{-1} \bigl( (J -iI) z_\l \bigr)^2}\\
		&= \bigl( \tfrac{2}{\pi \hbar} \bigr)^{n/2} \hspace{2pt} \det(F+I)^{-1/2} \hspace{2pt} e^{-\frac{1}{4 \hbar} z_\l^2}
		\hspace{2pt} e^{-\tfrac{1}{4 \hbar} (I+H)z_\l^2}.
	\end{align}
	The second identity is evident.
\end{proof}

This proves Theorem~\ref{thm_main2}.

\begin{corollary}
	If $F = I$, it is quickly checked that
	\begin{equation}
		(J-iI)^T (F+I)^{-1} (J-iI) = 0.
	\end{equation}
	In this case we get
	\begin{equation}
		((\Phi^\hbar, \widetilde{T}(z_\l) \Phi^\hbar)) = \left(\tfrac{1}{2\pi \hbar}\right)^{n/2} e^{-\tfrac{1}{4\hbar} z_\l^2} = (2\pi \hbar)^{n/2} A\varphi^\hbar(z_\l).
	\end{equation}
\end{corollary}

%

\section{Gaussian Mixed States on \texorpdfstring{$\R$}{R}}\label{sec_Gauss_1}

We treat the 1-dimensional case separately since this is the only case where for a Gaussian window $\varphi^\hbar$ a complete characterization of the frame set is known. Also, in this case all lattices are symplectic since $Sp(1) = SL(2,\R)$. We omit all proofs since the results follow already from Section \ref{sec_Gauss_n}. Note that the blocks $L$ and $P$ are now scalars. We denote the 1-dimensional standard Gaussian by
\begin{equation}\label{eq_standard_Gauss}
	\varphi^\hbar(x) = \left(\tfrac{1}{\pi \hbar}\right)^{1/4} \hspace{2pt} e^{-\tfrac{1}{2\hbar} x^2}.
\end{equation}
In time-frequency analysis $\hbar$ is usually set to $\tfrac{1}{2\pi}$. In this case we write $\varphi(x) = \varphi^{\tfrac{1}{2\pi}}(x) = 2^{1/4} e^{-\pi x^2}$. We note that $(\varphi^\hbar | \varphi^\hbar) = \norm{\varphi^\hbar}^2 = 1$ and the (general) frame set is given by
\begin{equation}
	\mathfrak{F}(\varphi^\hbar) = \{ \L \subset \R^2 \mid \delta_*(\L) > 2 \pi \hbar \}.
\end{equation}
Its Wigner transform is given by
\begin{equation}
	W\varphi^\hbar(x,p) = \tfrac{1}{\pi \hbar} \hspace{2pt} e^{- \tfrac{1}{\hbar} \left(x^2 + p^2\right)}.
\end{equation}
and its ambiguity function by
\begin{equation}\label{eq_ambi_Gauss}
	A\varphi^\hbar(x,p) = \tfrac{1}{2 \pi \hbar} \hspace{2pt} e^{- \tfrac{1}{4 \hbar} \left(x^2 + p^2\right)}.
\end{equation}
A generalized Gaussian in $\Lt[]$ is of the form
\begin{equation}\label{eq_generalized_Gauss}
	\widehat{M}_{L} \widehat{V}_{-P} \varphi^\hbar(x) = \varphi^\hbar_{P,L}(x) = \left(\tfrac{1}{\pi \hbar}\right)^{1/4} \sqrt{L} e^{-\tfrac{1}{2\hbar} (L^2 + iLP) x^2},
\end{equation}
where $\widehat{V}_{-P}$ is the metaplectic chirp and $\widehat{M}_{L}$ is the metaplectic dilation operator (where we ignore the Maslov index and set it to $m=0$) given by
\begin{equation}
	\widehat{V}_{-P} \psi(x) = e^{\tfrac{i}{2 \hbar} P x^2} \psi(x) \qquad \text{and} \qquad \widehat{M}_L \psi(x) = \sqrt{L} \hspace{2pt} \psi(L x), \quad L>0.
\end{equation}
For certain densities of the lattice, it is possible to explicitly compute Weyl--Heisenberg expansions of a function and the canonical dual window of $\varphi^\hbar$ using results from Baastians \cite{Baastians_1980} and Janssen \cite{Jan96}.
\begin{figure}[hp]
	\begin{center}
		\subfigure[The Gaussian $\varphi_{0,2}$.]{\includegraphics[width=.45\textwidth]{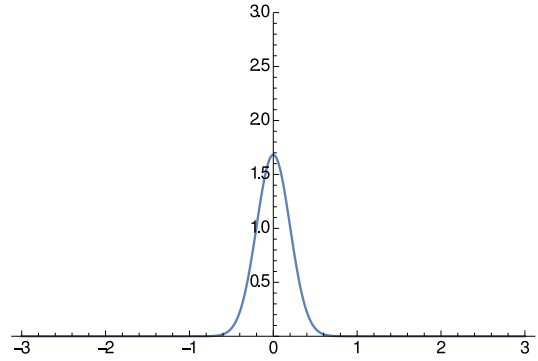}}
	\end{center}
	\subfigure[Approximation of $\varphi_{0,2}$ by using formula \eqref{eq_approx_expansion} for the integer lattice $\Z \times \Z$.]{\includegraphics[width=.45\textwidth]{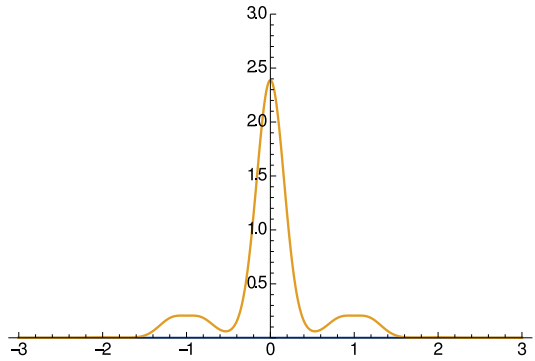}}
	\hfill
	\subfigure[Approximation of $\varphi_{0,2}$ by using formula \eqref{eq_approx_expansion} for the square lattice $\tfrac{1}{2}\Z \times \tfrac{1}{2}\Z$.]{\includegraphics[width=.45\textwidth]{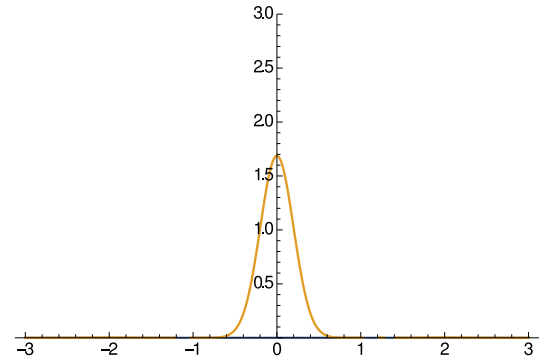}}
	\newline
	\subfigure[The pointwise difference (in absolute values) between the approximated Gaussian and the original Gaussian.]{\includegraphics[width=.45\textwidth]{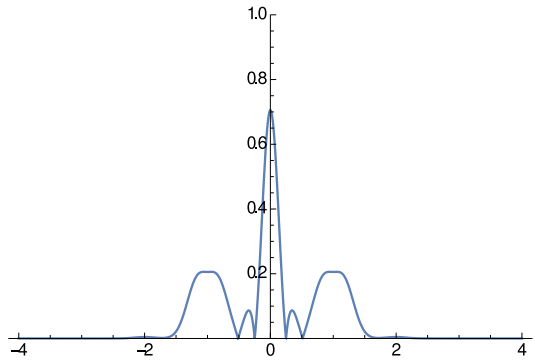}}
	\hfill
	\subfigure[The pointwise difference (in absolute values) between the approximated Gaussian and the original Gaussian.]{\includegraphics[width=.45\textwidth]{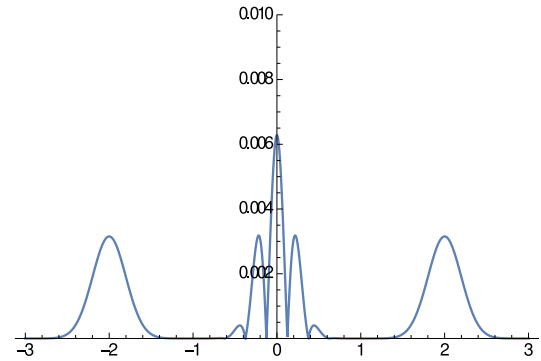}}
	\footnotesize{\caption{Comparison of a dilated Gaussian and its approximation by standard Gaussians using a truncated expansion of type \eqref{eq_approx_expansion} for square lattices at critical density $\delta = 1$ (left) and density $\delta = 4$ (right). For the plots we set $\hbar = \tfrac{1}{2 \pi}$. Note the different scales!}\label{fig_approx}}
\end{figure}

However, this is a quite cumbersome task. Since we know that the canonical dual window resembles (up to a factor determined by the density) the original window if the density is high enough, we might use the simpler approximation \eqref{eq_approx_expansion} instead of the Weyl--Heisenberg expansion \eqref{eq_frame_dual} (see Figure \ref{fig_approx} for an example). For a Gaussian window, this approximation becomes accurate quite quickly. In fact, we have the following property for the frame operator $\widehat{A}_\G$ associated to the Weyl--Heisenberg frame $\G(\varphi^\hbar, \L)$, where
\begin{equation}
	\L = \delta^{-1}
	\begin{pmatrix}
		L & 0\\
		0 & L^{-1}
	\end{pmatrix}
	\begin{pmatrix}
		1 & 0\\
		-P & 1
	\end{pmatrix} \Z^2.
\end{equation}
For our analysis we can focus on lattices generated by lower triangular matrices since any matrix can be decomposed into an orthogonal matrix and a lower triangular matrix (QR-decomposition). The orthogonal matrix can be ignored since $\varphi^\hbar$ is an eigenfunction with eigenvalue of modulus 1 of  the corresponding metaplectic operator (see, e.g., \cite{Faulhuber_Invariance_2016, ACHA}). In short, a rotation of the lattice does not affect our results since the Wigner transform and the ambiguity function of the standard Gaussian are invariant under rotations. Computing the ambiguity function of a generalized Gaussian yields
\begin{equation}
	A \varphi^\hbar_{P,L}(x,p) = \tfrac{1}{2 \pi \hbar} \hspace{2pt} e^{- \tfrac{1}{4 \hbar} \left(\left(L^2 + \tfrac{P^2}{L^2}\right) x^2 + 2 \tfrac{P}{L^2} \, xp + \tfrac{1}{L^2} p^2\right)}.
\end{equation}

In the 1-dimensional case, it is customary to parametrize a lattice in phase space by the triple $(\alpha, \beta, \gamma)$ in the following way
\begin{equation}
	\L =
	\begin{pmatrix}
		\alpha & 0\\
		\beta \gamma & \beta
	\end{pmatrix}
	\Z^2 =
	\begin{pmatrix}
		\alpha & 0\\
		0 & \beta
	\end{pmatrix}	
	\begin{pmatrix}
		1 & 0\\
		\gamma & 1
	\end{pmatrix}
	\Z^2 \, .
\end{equation}
Sometimes, the shearing and the dilation matrix are interchanged in the literature, however, this has no effect on our results. We also note that all 2-dimensional lattices can be expressed in the above form, up to a rotation which is negligible for the analysis of Gaussian states \cite{Faulhuber_Invariance_2016, ACHA}.

The following proposition is essentially \cite[Prop.~3.1.]{Faulhuber_Hexagon_2017}.
\begin{proposition}\label{pro_theta_estimate}
	For $r > 0$ and $\alpha, \, \beta, \, \gamma \in \R$ with $\alpha \beta = 1$ we have
	\begin{equation}
		\sum_{k,l \in \Z} e^{-r \left(\alpha^2 k^2 + 2 \alpha \beta \gamma \, k l + \beta^2 (1+\gamma)^2 l^2\right)}
		\leq \sum_{k,l \in \Z} e^{-r \left(\tfrac{1}{\alpha^2} k^2 + \tfrac{1}{\beta^2} l^2\right)}.
	\end{equation}
\end{proposition}
This yields the following result.

\begin{proposition}
	For $\G(\varphi^\hbar,\L)$ with $\L = \delta^{-1} Q
	\begin{pmatrix}
		\alpha & 0\\
		\beta \gamma & \beta
	\end{pmatrix} \Z^2$, where $\delta > 0$, $Q \in SO(2,\R)$ is an orthogonal matrix and $\alpha \beta = 1$, we have
	\begin{equation}
		\norm{Id - \delta^{-2} \widehat{A}_\G}_{op} = \mathcal{O}\left(e^{-\tfrac{\delta^2}{4 \hbar}\left(\tfrac{1}{\alpha^2} + \tfrac{1}{\beta^2}\right)}\right).
	\end{equation}
\end{proposition}

\begin{appendix}
	\section{The Fourier Transform of a Gaussian and Poisson's Summation Formula}\label{app_FT_Gauss}
	
	Recall that the Fourier transform of a function $\psi \in \mathcal{S}(\R^n)$ is given by
	\begin{equation}
		\F \psi(p) = \left(\tfrac{1}{2\pi \hbar}\right)^{n/2} \int_{\R^n} \psi(x) e^{-\tfrac{i}{\hbar} p  x} \, d^n x.
	\end{equation}
	The corresponding Plancherel's theorem reads
	\begin{equation}
		\norm{\psi}_{\Lt}^2 = \norm{\F\psi}_{\Lt}^2.
	\end{equation}
	For a Gaussian of type $\phi(x) = e^{-\tfrac{1}{2\hbar}Mx^{2}}$ such that M has positive definite real part and $M^* = M$, the Fourier transform is given by
	\begin{equation}
		\F\phi(p) = \left(\tfrac{1}{2\pi \hbar}\right)^{n/2} \int_{\R^n} e^{-\tfrac{1}{2\hbar}Mx^{2}} e^{-\tfrac{i}{\hbar} p  x} \, d^n x = \det(M)^{-1/2} e^{-\tfrac{1}{2\hbar} M^{-1} p^2}.
	\end{equation}
	This result is just a slight variation of Folland's formula \cite[App.~A, eq.~(1)]{Folland}. With these definitions Poisson's summation formula reads
	\begin{equation}
		\sum_{k \in \Z^n} \psi(k+x) = \sum_{l \in \Z^n} \F \psi(l) \hspace{2pt} e^{\tfrac{i}{\hbar} l x}.
	\end{equation}
	Since we will use the formula only for Gaussians, we omit the technical details for when this formula holds pointwise.

\end{appendix}

\begin{acknowledgement}
	The authors thank the anonymous referee for a careful reading of the manuscript and pointing out some references. Markus Faulhuber was supported by the Austrian Science Fund (FWF) project P27773-N23 and by the Erwin--Schrödinger Program of the Austrian Science Fund (FWF) J4100-N32. Maurice A.~de Gosson was supported by the Austrian Science Fund (FWF) projects P23902-N13 and P27773-N23. David Rottensteiner was supported by the Austrian Science Fund (FWF) project P27773-N23.
\end{acknowledgement}


\begin{thebibliography}{11pt}
	
	\bibitem{ARL}
	G.~Adesso, S.~Ragy and A.R.~Lee.
	\newblock{Continuous variable quantum information: Gaussian states and beyond}.
	\newblock{Open Systems \& Information Dynamics} 21(01n02):1440001, (2014).
	
	\bibitem{Baastians_1980}
	M.J.~Baastians.
	\newblock{Gabor's expansion of a signal into Gaussian elementary signals}.
	\newblock{\em Proceedings of the IEEE}, 68(4):538--539, (1980).
	
	\bibitem{AB2}
	A.~Bourouihiya.
	\newblock{The tensor product of frames}.
	\newblock{\em Sampling Theory in Signal \& Image Processing}, 7(1):65--76, (2008).
	
	\bibitem{goetal}
	N.C.~Dias, M.A.~de Gosson, F.~Luef, J.~Prata,
	\newblock{A pseudo-differential calculus on non-standard symplectic space; spectral and regularity results in modulation spaces}.
	\newblock{\em Journal de mathematiques pures et appliquees}, 96(5):423--445, (2011).
	
	\bibitem{Faulhuber_Invariance_2016}
	M.~Faulhuber.
	\newblock {Gabor frame sets of invariance: a Hamiltonian approach to Gabor frame deformations}.
	\newblock {\em Journal of Pseudo-Differential Operators and Applications}, 7(2):213--235, (2016).
	
	\bibitem{Faulhuber_PhD}
	M.~Faulhuber.
	\newblock{\em Extremal Bounds of Gaussian Gabor Frames and Properties of Jacobi's Theta Functions}.
	\newblock{Doctoral Thesis}, University of Vienna, (2016).
	
	\bibitem{Faulhuber_Hexagon_2017}
	M.~Faulhuber
	\newblock{Minimal Frame Operator Norms via Minimal Theta Functions}.
	\newblock{\em Journal of Fourier Analysis and Applications}, 24(2):545--559, (2018).
	
	\bibitem{Faulhuber_Steinerberger_2017}
	M.~Faulhuber, S.~Steinerberger
	\newblock{Optimal Gabor frame bounds for separable lattices and estimates for Jacobi theta functions}.
	\newblock{\em Journal of Mathematical Analysis and Applications}, 445(1):407--422, (2017).
	
	\bibitem{Fei81}
	H.G.~Feichtinger.
	\newblock{On a new Segal algebra}.
	\newblock{\em Monatshefte für Mathematik}, 92(4):269--289, (1981).
	
	\bibitem{FeiZim98}
	H.G.~Feichtinger, G.~Zimmermann.
	\newblock {A Banach space of test functions in Gabor analysis}.
	\newblock In H.G.~Feichtinger and T.~Strohmer, eds., {\em Gabor Analysis and Algorithms: Theory and Applications}, pp.~123--170, Birkh{\"a}user, (1998). 
	
	\bibitem{Folland}
	G.B.~Folland.
	\newblock {\em {Harmonic analysis in phase space}}.
	\newblock Number 122 in {Annals of Mathematics Studies}. Princeton University Press, (1989).
	
%
	\bibitem{ACHA}
	M.A.~de Gosson.
	\newblock{Hamiltonian deformations of Gabor frames: First steps}.
	\newblock{\em Applied and Computational Harmonic Analysis}, 38(2):196--221, (2014).
	
	\bibitem{gobopp1}
	M.A.~de Gosson.
	\newblock{Phase Space Weyl Calculus and Global Hypoellipticity of a Class of Degenerate Elliptic Partial Differential Operators}.
	\newblock In L.~Rodino and M.W.~Wong, eds., {\em New Developments in Pseudo-Differential Operators}, volume 189 of {\em Operator Theory: Advances and Applications}, pp.1--14, Birkh\"auser, (2009).
	
	\bibitem{gobopp2}
	M.A.~de Gosson.
	\newblock{Spectral properties of a class of generalized Landau operators}.
	\newblock{\em Communications in Partial Differential Equations}, 33(11):2096--2104, (2008).
	
	\bibitem{Birk}
	M.A.~de Gosson.
	\newblock{\em Symplectic Geometry and Quantum Mechanics}.
	\newblock{Birkh\"{a}user}, (2006).

	\bibitem{Birkbis}
	M.A.~de Gosson.
	\newblock{\em Symplectic Methods in Harmonic Analysis and in Mathematical Physics}.
	\newblock{Birkh\"{a}user}, (2011).
	
	\bibitem {JGP17}
	M.A.~de Gosson.
	\newblock{The canonical group of transformations of a Weyl--Heisenberg frame; applications to Gaussian and Hermitian frames}.
	\newblock{\em Journal of Geometry and Physics}, 114:375--383 (2017).
	
	\bibitem {gowig}
	M.A.~de Gosson.
	\newblock{\em The Wigner Transform}.
	\newblock{World Scientific}, Singapore, (2017).
	
	\bibitem {golubopp}
	M.A.~de Gosson, F.~Luef.
	\newblock{Born--Jordan Pseudodifferential Calculus, Bopp Operators and Deformation Quantization}.
	\newblock{\em Integral Equations and Operator Theory}, 84(4):463--485, (2016).
	
	
	\bibitem{Gro01}
	K.~Gr{\"o}chenig.
	\newblock {\em Foundations of Time-Frequency Analysis}.
	\newblock {Applied and Numerical Harmonic Analysis}. Birkh{\"a}user, Boston, MA, (2001).
	
	\bibitem{Gro14}
	K.~Gr{\"o}chenig.
	\newblock{The mystery of Gabor frames.}
  	\newblock{\em Journal of Fourier Analysis and Applications}, 20(4):865--895, (2014).	
	
	\bibitem{Hei07}
	Heil, Christopher.
	\newblock{History and evolution of the density theorem for Gabor frames.}
	\newblock{\em Journal of Fourier Analysis and Applications}, 13(2):113--166, (2007).	

	\bibitem{Horo_1996}
	M.~Horodecki, P.~Horodecki, R.~Horodecki.
	\newblock{Separability of mixed states: necessary and sufficient conditions}.
	\newblock{\em Physics Letters A}, 223(1):1--8, (1996).
	
	\bibitem{Jak18}
	M.S.~Jakobsen.
	\newblock{On a (No Longer) New Segal Algebra: A Review of the Feichtinger Algebra.}
	\newblock{\em Journal of Fourier Analysis and Applications}, pp.~1--82, (2018).
	
	\bibitem{Janssen_Duality_1994}
	A.J.E.M.~Janssen.
	\newblock{Duality and Biorthogonality for Weyl-Heisenberg Frames}.
	\newblock{\em Journal of Fourier Analysis and Applications},4(1):403--436, (1994).
	
	\bibitem{Jan96}
	A.J.E.M.~Janssen.
	\newblock {Some Weyl-Heisenberg frame bound calculations}.
	\newblock {\em Indagationes Mathematicae}, 7(2):165--183, (1996).
	
	\bibitem{Lein_2006}
	J.M.~Leinaas, J.~Myrheim, E.~Ovrum
	\newblock{Geometrical aspects of entanglement}.
	\newblock{\em Physical Review A} 74:012313, (2006).
	
	\bibitem{Littlejohn}
	R.G.~Littlejohn.
	\newblock {The semiclassical evolution of wave packets}.
	\newblock {\em Physics Reports}, 138(4--5):193--291, (1986).
	
	\bibitem{Lyu92}
	Y.~Lyubarskii.
	\newblock {Frames in the Bargmann space of entire functions}.
	\newblock In {\em {Entire and Subharmonic Functions}}, pp.~167--180. American Mathematical Society, Providence, RI, (1992).
	
	
	
	\bibitem{goetz}
	G.E.~Pfander, P.~Rashkov, Y.~Wang.
	\newblock{Geometric Construction of Tight Multivariate Gabor Frames with Compactly Supported Smooth Windows}.
	\newblock{\em Journal of Fourier Analysis and Applications}, 18(2):223--239, (2012).
	
	\bibitem{Sei92}
	K.~Seip.
	\newblock {Density theorems for sampling and interpolation in the Bargmann-Fock space}.
	\newblock {\em American Mathematical Society. Bulletin. New Series}, 26(2):322--328, (1992).
	
	\bibitem{sewa92}
	K.~Seip , R.~Wallst\'{e}n,
	\newblock{Density theorems for sampling and interpolation in the Bargmann--Fock space. II},
	\newblock{\em Journal f\"ur die reine und angewandte Mathematik}, 1992(429):107--113, (1992).
	
	\bibitem{se17}
	K.P.~Seshadreesan, L.~Lami and M.M.~Wilde.
	\newblock{R\'{e}nyi relative entropies of quantum Gaussian states}.
	\newblock{\em arXiv preprint}, arXiv:1706.09885, (2017).
	
	
\end{thebibliography}

\end{document}